\documentclass{article}%
\usepackage{amsmath}
\usepackage{amsfonts}
\usepackage{amssymb}
\usepackage{graphicx}%
\providecommand{\U}[1]{\protect\rule{.1in}{.1in}}
\newtheorem{theorem}{Theorem} [section]

\newtheorem{conjecture}[theorem]{Conjecture}

\newenvironment{proof}[1][Proof]{\noindent\textbf{#1.} }{\ \rule{0.5em}{0.5em}}
\begin{document}

\author{Vadim E. Levit and David Tankus\\Department of Computer Science and Mathematics\\Ariel University Center of Samaria, ISRAEL\\\{levitv, davidta\}@ariel.ac.il}
\title{On Relating Edges in Graphs without Cycles of Length $4$}
\date{}
\maketitle

\begin{abstract}
An edge $xy$ is \textit{relating} in the graph $G$ if there is an independent
set $S$, containing neither $x$ nor $y$, such that $S\cup\{x\}$ and
$S\cup\{y\}$ are both maximal independent sets in $G$. It is an \textbf{NP}%
-complete problem to decide whether an edge is relating \cite{bnz:related}. We
show that the problem remains \textbf{NP}-complete even for graphs without
cycles of length $4$ and $5$. On the other hand, for graphs without cycles of
length $4$ and $6$, the problem can be solved in polynomial time.

\end{abstract}

\section{Introduction}

Throughout this paper $G=(V,E)$ is a simple (i.e., a finite, undirected,
loopless and without multiple edges) graph with vertex set $V=V(G)$ and edge
set $E=E(G).$

Let $S\subseteq V$ be a set of vertices, and let $i\in\mathbb{N}$. Then
\[
N_{i}(S)=\{w\in V|\ min_{s\in S}\ d(w,s)=i\},
\]
where $d(x,y)$ is the minimal number of edges required to construct a path
between $x$ and $y$. If $i\neq j$ then $N_{i}(S)\cap N_{j}(S)=\phi$. If
$S=\{v\}$ for some $v\in V$, then $N_{i}(\{v\})$ is abbreviated to $N_{i}(v)$.

A set of vertices $S\subseteq V$ is \textit{independent }if for every $x,y\in
S$, $x$ and $y$ are not adjacent. It is clear that an empty set is
independent. The \textit{independence number} of $G$, denoted by $\alpha(G)$,
is the cardinality of the maximum size independent set in the graph.

A graph is \textit{well-covered} if every maximal independent set has the same
cardinality, $\alpha(G)$.

Let $T\subseteq V$. Then $S$ \textit{dominates} $T$ if $S\cup N_{1}%
(S)\supseteq T$. If $S$ and $T$ are both empty, then $N_{1}(S)=\phi$, and
therefore $S$ dominates $T$. If $S$ is a maximal independent set of $G$, then
it dominates the whole graph.

Two adjacent vertices, $x$ and $y$, in $G$ are said to be \textit{related} if
there is an independent set $S$, containing neither $x$ nor $y$, such that
$S\cup\{x\}$ and $S\cup\{y\}$ are both maximal independent sets in the graph.
If $x$ and $y$ are related, then $xy$ is a \textit{relating edge}. To decide
whether an edge in an input graph is relating is an \textbf{NP}-complete
problem \cite{bnz:related}.

\begin{theorem}
\label{relatedNPC} \cite{bnz:related} The following problem is \textbf{NP}-complete:

Input: A graph $G=(V,E)$, and an edge $xy\in E$.

Question: Is $xy$ a relating edge in $G$?
\end{theorem}

In \cite{bnz:related}, Brown, Nowakowski and Zverovich investigate
well-covered graphs with no cycles of length $4$. They denote the set of such
graphs by $\mathcal{WC}$$(\widehat{C}_{4})$, and prove the following.

\begin{theorem}
\label{remove} \cite{bnz:related} Let $G\in\mathcal{WC}(\widehat{C}_{4})$. If
$xy$ is an edge in $G$, but $x$ and $y$ are not related, then $G-xy$ is
well-covered and $\alpha(G)=\alpha(G-xy)$.
\end{theorem}

In this paper we continue the investigation of the structure of graphs with no
cycles of length $4$. We denote the set of graphs without cycles of sizes $k$
and $l$ by $\mathcal{G}(\widehat{C}_{k},\widehat{C}_{l})$. We prove that
Theorem \ref{relatedNPC} holds even for the case, where the input graph does
not contain cycles of length $4$ and $5$, i.e., $G\in\mathcal{G}(\widehat
{C}_{4},\widehat{C}_{5})$. On the other hand, if the input graph does not
contain cycles of length $4$ and $6$, i.e., $G\in\mathcal{G}(\widehat{C}%
_{4},\widehat{C}_{6})$, then the problem of identifying relating edges turns
out to be polynomial.

The fact that identifying relating edges is \textbf{NP}-complete for the input
restricted to $\mathcal{G}(\widehat{C}_{4},\widehat{C}_{5})$ is important,
because the analogous problem concerning well-covered graphs is known to be
polynomial \cite{fhn:wc45}.

\begin{theorem}
\label{wcc4c5npc} \cite{fhn:wc45} The following problem can be solved in
polynomial time:

Input: A graph $G\in\mathcal{G}(\widehat{C}_{4},\widehat{C}_{5})$.

Question: Is $G$ well-covered?
\end{theorem}

\section{Main Results}

Let $X=\{x_{1},...,x_{n}\}$ be a set of 0-1 variables. We define the set of
\textit{literals} $L_{X}$ over $X$ by $L_{X} = \{x_{i}, \overline{x_{i}} : i =
1,...,n\}$, where $\overline{x} = 1 - x$ is the \textit{negation} of $x$. A
\textit{truth assignment} to $X$ is a mapping $t:X \longrightarrow\{0,1\}$
that assigns a value $t(x_{i}) \in\{0,1\}$ to each variable $x_{i} \in X$. We
extend $t$ to $L_{X}$ by putting $t(\overline{x_{i}}) = \overline{t(x_{i})}$.
A literal $l \in L_{X}$ is true under $t$ if $t(l) = 1$. A \textit{clause}
over $X$ is a conjunction of some literals of $L_{X}$. Let $C=\{c_{1}%
,...,c_{m}\}$ be a set of clauses over $X$. A truth assignment $t$ to $X$
\textit{satisfies} a clause $c_{j} \in C$ if $c_{j}$ involves at least one
true literal under $t$.

SAT is a well-known \textbf{NP}-complete problem \cite{gj:NPC}. It is defined
as follows.\newline\textit{Input}: A set of variables $X=\{x_{1},...,x_{n}\}$,
and a set of clauses $C=\{c_{1},...,c_{m}\}$ over $X$.\newline%
\textit{Question}: Is there a truth assignment to $X$ which satisfies all
clauses of $C$?

\begin{theorem}
\label{related4NPC} The following problem is \textbf{NP}-complete:

Input: A graph $G=(V,E)\in\mathcal{G}(\widehat{C}_{4},\widehat{C}_{5})$, and
an edge $xy\in E$.

Question: Is $xy$ a relating edge in $G$?
\end{theorem}

\begin{proof}
Clearly, the problem is in \textbf{NP}. We use a polynomial time reduction
from SAT. Let $(X=\{x_{1},...,x_{n}\},C=\{c_{1},...,c_{m}\})$ be an instance
of SAT. We construct a graph $G=G_{X,C}$ as follows (see Figure \ref{SAT45}).

The vertex set of $G$ contains:

\begin{itemize}
\item Two vertices, $x$ and $y$.

\item A set $T = \{ x_{i}, t_{i}, f_{i} : 1 \leq i \leq n \}$.

\item A set $C = \{ c_{j} : 1 \leq j \leq m \}$.

\item A set $I_{T} = \{t_{i,j} : 1 \leq i \leq n, 1 \leq j \leq m, x_{i}
\ appears \ in \ c_{j}\}$.

\item A set $I_{F} = \{f_{i,j} : 1 \leq i \leq n, 1 \leq j \leq m,
\overline{x_{i}} \ appears \ in \ c_{j} \}$.
\end{itemize}

The edge set of $G$ contains:

\begin{itemize}
\item The edge $xy$.

\item All edges $yx_{i}$, for $1 \leq i \leq n$.

\item All triangles $( x_{i}, t_{i}, f_{i} )$, for $1 \leq i \leq n$.

\item An edge $t_{i}f_{i,j}$, if $x_{i}$ appears in $c_{j}$, for $1 \leq i
\leq n$ and $1 \leq j \leq m$.

\item An edge $f_{i}t_{i,j}$, if $\overline{x_{i}}$ appears in $c_{j}$, for $1
\leq i \leq n$ and $1 \leq j \leq m$.

\item An edge $t_{i,j}c_{j}$, if $x_{i}$ appears in $c_{j}$, for $1 \leq i
\leq n$ and $1 \leq j \leq m$.

\item An edge $f_{i,j}c_{j}$, if $\overline{x_{i}}$ appears in $c_{j}$, for $1
\leq i \leq n$ and $1 \leq j \leq m$.

\item All edges $xc_{j}$, for $1\leq j\leq m$.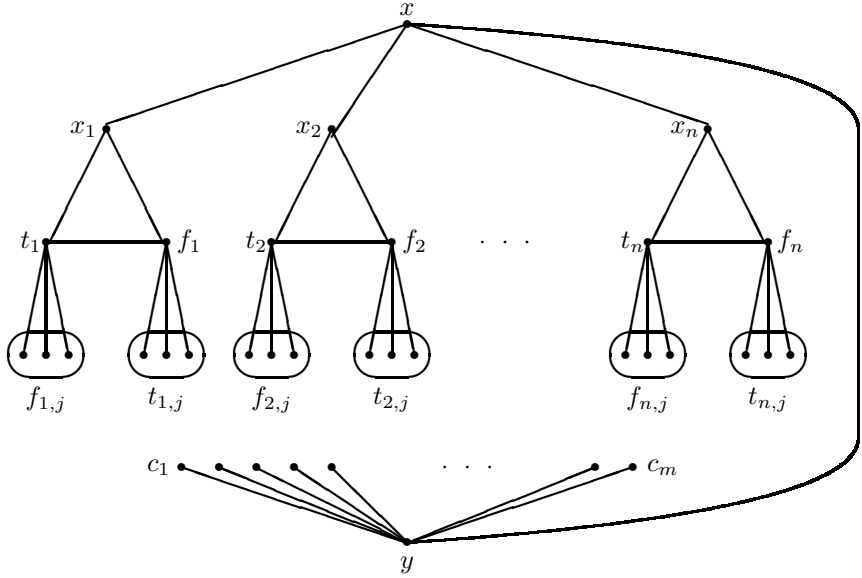
\begin{figure}[h]
\setlength{\unitlength}{1.0cm} \begin{picture}(20,8)\thicklines
\put(7,7.4){\circle*{0.1}}
\put(3,6){\circle*{0.1}}
\put(6,6){\circle*{0.1}}
\put(11,6){\circle*{0.1}}
\put(2.2,4.5){\circle*{0.1}}
\put(3.8,4.5){\circle*{0.1}}
\put(5.2,4.5){\circle*{0.1}}
\put(6.8,4.5){\circle*{0.1}}
\put(10.2,4.5){\circle*{0.1}}
\put(11.8,4.5){\circle*{0.1}}
\multiput(1.9,3)(0.3,0){3}{\circle*{0.1}}
\multiput(3.5,3)(0.3,0){3}{\circle*{0.1}}
\multiput(4.9,3)(0.3,0){3}{\circle*{0.1}}
\multiput(6.5,3)(0.3,0){3}{\circle*{0.1}}
\multiput(9.9,3)(0.3,0){3}{\circle*{0.1}}
\multiput(11.5,3)(0.3,0){3}{\circle*{0.1}}
\multiput(4,1.5)(0.5,0){5}{\circle*{0.1}}
\multiput(9.5,1.5)(0.5,0){2}{\circle*{0.1}}
\put(7,0.5){\circle*{0.1}}
\put(7,7.6){\makebox(0,0){$x$}}
\put(2.7,6){\makebox(0,0){$x_{1}$}}
\put(5.7,6){\makebox(0,0){$x_{2}$}}
\put(10.7,6){\makebox(0,0){$x_{n}$}}
\put(2,4.5){\makebox(0,0){$t_{1}$}}
\put(5,4.5){\makebox(0,0){$t_{2}$}}
\put(10,4.5){\makebox(0,0){$t_{n}$}}
\put(4.1,4.5){\makebox(0,0){$f_{1}$}}
\put(7.1,4.5){\makebox(0,0){$f_{2}$}}
\put(12.1,4.5){\makebox(0,0){$f_{n}$}}
\put(7,0.2){\makebox(0,0){$y$}}
\put(7,7.4){\line(-3,-1){4}}
\put(7,7.4){\line(-2,-3){1}}
\put(7,7.4){\line(3,-1){4}}
\put(3,6){\line(-1,-2){0.75}}
\put(3,6){\line(1,-2){0.75}}
\put(2.2,4.5){\line(2,0){1.6}}
\put(6,6){\line(-1,-2){0.75}}
\put(6,6){\line(1,-2){0.75}}
\put(5.2,4.5){\line(2,0){1.6}}
\put(11,6){\line(-1,-2){0.75}}
\put(11,6){\line(1,-2){0.75}}
\put(10.2,4.5){\line(2,0){1.6}}
\put(2.2,4.5){\line(-1,-5){0.3}}
\put(2.2,4.5){\line(0,-1){1.5}}
\put(2.2,4.5){\line(1,-5){0.3}}
\put(5.2,4.5){\line(-1,-5){0.3}}
\put(5.2,4.5){\line(0,-1){1.5}}
\put(5.2,4.5){\line(1,-5){0.3}}
\put(10.2,4.5){\line(-1,-5){0.3}}
\put(10.2,4.5){\line(0,-1){1.5}}
\put(10.2,4.5){\line(1,-5){0.3}}
\put(3.8,4.5){\line(-1,-5){0.3}}
\put(3.8,4.5){\line(0,-1){1.5}}
\put(3.8,4.5){\line(1,-5){0.3}}
\put(6.8,4.5){\line(-1,-5){0.3}}
\put(6.8,4.5){\line(0,-1){1.5}}
\put(6.8,4.5){\line(1,-5){0.3}}
\put(11.8,4.5){\line(-1,-5){0.3}}
\put(11.8,4.5){\line(0,-1){1.5}}
\put(11.8,4.5){\line(1,-5){0.3}}
\put(7,0.5){\line(-1,1){1}}
\put(7,0.5){\line(-3,2){1.5}}
\put(7,0.5){\line(-5,2){2.5}}
\put(7,0.5){\line(-2,1){2}}
\put(7,0.5){\line(-3,1){3}}
\put(7,0.5){\line(5,2){2.5}}
\put(7,0.5){\line(3,1){3}}
\put(3.7,1.5){\makebox(0,0){$c_{1}$}}
\put(10.4,1.5){\makebox(0,0){$c_{m}$}}
\qbezier(7,7.4)(13,7)(13,6)
\qbezier(13,6)(13,4)(13,1.9)
\qbezier(7,0.5)(13,0.9)(13,1.9)
\put(2.2, 3){\oval(1, 0.6)}
\put(3.8, 3){\oval(1, 0.6)}
\put(5.2, 3){\oval(1, 0.6)}
\put(6.8, 3){\oval(1, 0.6)}
\put(10.2, 3){\oval(1, 0.6)}
\put(11.8, 3){\oval(1, 0.6)}
\put(2.2,2.4){\makebox(0,0){$f_{1,j}$}}
\put(3.8,2.4){\makebox(0,0){$t_{1,j}$}}
\put(5.2,2.4){\makebox(0,0){$f_{2,j}$}}
\put(6.8,2.4){\makebox(0,0){$t_{2,j}$}}
\put(10.2,2.4){\makebox(0,0){$f_{n,j}$}}
\put(11.8,2.4){\makebox(0,0){$t_{n,j}$}}
\multiput(8,4.5)(0.3,0){3}{\circle*{0.05}}
\multiput(7.5,1.5)(0.3,0){3}{\circle*{0.05}}
\end{picture}
\caption{The structure of the graph $G_{X,C}$.}%
\label{SAT45}%
\end{figure}
\end{itemize}

The graph $G$ does not contain cycles of length 4 and 5. We show that $xy$ is
a relating edge in $G$ if and only if $(X,C)$ has a satisfying truth assignment.

Let $\Phi$ be a satisfying truth assignment for $(X,C)$. Define $S = \{ t_{i},
t_{i,j}: \Phi( x_{i} ) = 1 \} \cup\{ f_{i}, f_{i,j}: \Phi( x_{i} ) = 0 \}$.
Clearly, $S$ is independent. The fact that $\Phi$ is a satisfying truth
assignment implies that $S$ dominates $C$. Hence, $S \cup\{x\}$ and $S
\cup\{y\}$ are both maximal independent sets in $G$, and $xy$ is a relating edge.

Conversely, assume $xy$ is a relating edge. Let $S$ be an independent set,
such that $S \cup\{x\}$ and $S \cup\{y\}$ are both maximal independent sets in
$G$. Clearly, $S$ does not contain vertices of $C \cup\{ x_{1},...,x_{n} \}$.
Hence, for each $1 \leq i \leq n$ exactly one of $t_{i}$ and $f_{i}$ belongs
to $S$. If $t_{i} \in S$ then $t_{i,j} \in S$ for each possible $j$. If $f_{i}
\in S$ then $f_{i,j} \in S$ for each possible $j$. Define a truth assignment
$\Phi$: If $t_{i} \in S$ then $\Phi(x_{i})=1$, else $\Phi(x_{i})=0$, for every
$1 \leq i \leq n$. The fact that $C$ is dominated by $S$ implies that every
clause of $C$ involves a true literal. Therefore, $\Phi$ is a satisfying truth
assignment for $(X,C)$.
\end{proof}

\begin{theorem}
\label{new} The following problem can be solved in polynomial time:

Input: A graph $G=(V,E)\in\mathcal{G}(\widehat{C}_{4},\widehat{C}_{6})$, and
an edge $xy\in E$.

Question: Is $xy$ a relating edge in $G$?
\end{theorem}

\begin{proof}
For every $v\in\{x,y\}$, let $u=\{x,y\}-\{v\}$, and define: \newline%
$M_{1}(v)=N_{1}(v)\cap N_{2}(u)$, $M_{2}(v)=N_{1}(M_{1}(v))-\{v\}$.

The vertices $x$ and $y$ are related if and only if there exists an
independent set in $M_{2}(x)\cup M_{2}(y)$ which dominates $M_{1}(x)\cup
M_{1}(y)$.

The fact that the graph does not contain cycles of length $6$ implies the
following 3 conclusions:

\begin{itemize}
\item There are no edges which connect vertices of $M_{2}(x)$ with vertices of
$M_{2}(y)$.

\item The set $M_{2}(x) \cap M_{2}(y)$ is independent.

\item There \ are \ no \ edges \ between \ $M_{2}(x) \cap M_{2}(y)$ \ and
\ other \ vertices \ of $M_{2}(x) \cup M_{2}(y)$.
\end{itemize}

Hence, if $S_{x}\subseteq M_{2}(x)$ and $S_{y}\subseteq M_{2}(y)$ are
independent, then $S_{x}\cup S_{y}$ is independent, as well. Therefore, it is
enough to prove that one can decide in polynomial time whether there exists an
independent set in $M_{2}(v)$ which dominates $M_{1}(v)$, where $v\in\{x,y\}$.

Let $v$ be any vertex in $\{x,y\}$. Every vertex of $M_{2}(v)$ is adjacent to
exactly one vertex of $M_{1}(v)$, or otherwise the graph contains a $C_{4}$.
Every connectivity component of $M_{2}(v)$ contains at most $2$ vertices, or
otherwise the graph contains either a $C_{4}$ or a $C_{6}$. Let $A_{1}%
,...,A_{k}$ be the connectivity components of $M_{2}(v)$.

Define a flow network $F_{v}=\{G_{F}=(V_{F},E_{F}),s\in V_{F},t\in
V_{F},w:E_{F}\longrightarrow R\}$ as follows.

Let $V_{F}=M_{1}(v)\cup M_{2}(v)\cup\{a_{1},...,a_{k},s,t\}$, where
$a_{1},...,a_{k},s,t$ are new vertices, $s$ and $t$ are the source and sink of
the network, respectively.

The directed edges $E_{F}$ are:

\begin{itemize}
\item the directed edges from $s$ to each vertex of $M_{1}(v)$;

\item all directed edges $v_{1}v_{2}$ s.t. $v_{1}\in M_{1}(v)$, $v_{2}\in
M_{2}(v)$ and $v_{1}v_{2}\in E$;

\item the directed edges $va_{i}$, for each $1\leq i\leq k$ and for each $v\in
A_{i}$;

\item the directed edges $a_{i}t$, for each $1\leq i\leq k$.
\end{itemize}

Let $w\equiv1$. Invoke any polynomial time algorithm for finding a maximum
flow in the network, for example Ford and Fulkerson's algorithm. Let $S_{v}$
be the set of vertices in $M_{2}(v)$ in which there is a positive flow.
Clearly, $S_{v}$ is independent. The maximality of $S_{v}$ implies that
$|M_{1}(v)\cap N_{1}(S_{v})|\geq|M_{1}(v)\cap N_{1}(S_{v}^{\prime})|$, for any
independent set $S_{v}^{\prime}$ of $M_{2}(v)$.

Let us conclude the proof with the recognition algorithm for relating edges.

For each $v\in\{x,y\}$, build a flow network $F_{v}$ as described above, and
find a maximum flow. Let $S_{v}$ be the set of vertices in $M_{2}(v)$ in which
there is a positive flow. If $S_{v}$ does not dominate $M_{1}(v)$ the
algorithm terminates announcing that $x$ and $y$ are not related. Otherwise,
let $S$ be any maximal independent set of $G-\{x,y\}$ which contains
$S_{x}\cup S_{y}$. Each of $S\cup\{x\}$ and $S\cup\{y\}$ is a maximal
independent set of $G$, and $x$, $y$ are related.

This algorithm can be implemented in polynomial time: One iteration of Ford
and Fulkerson's algorithm includes:

\begin{itemize}
\item Updating the flow function. (In the first iteration the flow is equal to
$0$.)

\item Constructing the residual graph.

\item Finding an augmenting path, if exists. It is worth mentioning that the
residual capacity of every augmenting path equals $1$.
\end{itemize}

Each of the above can be implemented in $O\left(  \left\vert V\right\vert
+\left\vert E\right\vert \right)  $ time. In each iteration the number of
vertices in $M_{2}(v)$ with a positive flow increases by $1$. Therefore, the
number of iterations can not exceed $\left\vert V\right\vert $, and Ford and
Fulkerson's algorithm terminates in $O\left(  \left\vert V\right\vert \left(
\left\vert V\right\vert +\left\vert E\right\vert \right)  \right)  $ time. Our
algorithm invokes Ford and Fulkerson's algorithm twice, and terminates in
$O\left(  \left\vert V\right\vert \left(  \left\vert V\right\vert +\left\vert
E\right\vert \right)  \right)  $ time.
\end{proof}

\section{Conjectures}

Our main conjecture reads as follows.

\begin{conjecture}
For every integer $k\geq7$, the following recognition problem is \textbf{NP}-complete.

Input: A graph $G=(V,E)\in\mathcal{G}(\widehat{C}_{4},\widehat{C}_{k})$, and
an edge $xy\in E$.

Question: Is $xy$ a relating edge in $G$?
\end{conjecture}

\end{document}